\documentclass[letterpaper, 10 pt, journal, romanappendices]{IEEEtran}

\usepackage{amsmath, amssymb}
\usepackage{graphicx}
\usepackage[sans]{dsfont}
\usepackage{stmaryrd}
\usepackage{mathrsfs}
\usepackage{subfigure}
\usepackage{color}
\usepackage{cite}
\usepackage{asymptote}

\DeclareMathAlphabet{\mpzc}{OT1}{pzc}{m}{it}

\newcommand{\msf}{\mathsf}

\newcommand{\mbb}{\mathbb}

\newcommand{\mcl}{\mathcal}
\newcommand{\msc}{\mathscr}
\newcommand{\mds}{\mathds}
\newcommand{\emp}{\emph}

\newcommand{\tends}{\rightarrow}

\newcommand{\e}{\varepsilon}

\newcommand{\p}{\msf{p}}
\newcommand{\msi}{\msf{i}}
\newcommand{\msd}{\msf{d}}

\newcommand{\bb}{\overline{\beta}}

\newtheorem{thm}{Theorem}

\newtheorem{prop}[thm]{Proposition}
\newtheorem{cor}[thm]{Corollary}

\title{Write Channel Model for\\Bit-Patterned Media Recording}
\author{Aravind~R.~Iyengar, Paul~H.~Siegel,~\IEEEmembership{Fellow,~IEEE} and~Jack~K.~Wolf,~\IEEEmembership{Life~Fellow,~IEEE}% <-this % stops a space
\thanks{The authors are with the Department of Electrical and Computer Engineering and the Center for Magnetic Recording Research, University of California, San Diego, La Jolla, CA 92093 USA (e-mail: aravind@ucsd.edu, psiegel@ucsd.edu, jwolf@ucsd.edu).}
\thanks{This work was supported in part by the Center for Magnetic Recording Research and by the National Science Foundation under the Grant CCF-$0829865$.}
\thanks{Portions of this work were presented at the IEEE International Symposium on Information Theory, Austin, TX, June 13--18, 2010 \cite{iye_10_isit_wcm}. This paper has been accepted for publication in the IEEE Transactions on Magnetics.}}

\newcommand{\twobibs}[2]{#2}

\begin{document}
\maketitle

\begin{abstract}
We propose a new write channel model for bit-patterned media recording that reflects the data dependence of write synchronization errors. It is shown that this model accommodates both substitution-like errors and insertion-deletion errors whose  statistics are determined by an underlying channel state process. We study information theoretic properties of the write channel model, including the capacity, symmetric information rate, Markov-$1$ rate and the zero-error capacity.
\end{abstract}


\section{Introduction}
\PARstart M{agnetic} recording channels are typically modeled as binary-input intersymbol interference (ISI) channels, also called \emp{partial response} channels \cite{kob_71_tcomtech_codmagrec, kab_75_tcom_prs}. An implicit assumption made in these channel models is that the data are correctly written on the disk and that errors occur only during the readback process. While this is a realistic assumption in conventional recording on continuous media, it is questionable in the context of certain advanced recording technologies, such as bit-patterned media (BPM) recording, that may be required to achieve higher storage densities. In this paper, we will examine some of the underlying causes of errors in the recording process, particularly in BPM recording, and then propose a new probabilistic write channel model that captures some of the data dependence of these write errors. Thus, the input to this write channel model is the data sequence to be recorded and the channel output is the ``noisy'' sequence that actually gets stored on the medium.   This leads to a description of the full data recording and readback process as a cascade of an imperfect write channel and a noisy (partial response) readback channel. 

Hu et al. \cite{hu_07_tmag_bpmwc} proposed a model for the BPM recording channel in which the write channel was a  binary symmetric channel (BSC), and the readback channel was a linear, intersymbol-interference (ISI) channel with additive noise. They proposed and evaluated detection methods for this channel, and they investigate achievable information rates for such a system. In \cite{iye_09_all_bpm}, we considered an idealized cascaded channel model in which the write channel was again a BSC and the readback channel was a memoryless, binary-input, additive white Gaussian noise (AWGN) channel.  We studied theoretical properties of this channel, as well as the decodable regions of LDPC codes under a number of decoding algorthms. In this paper, we focus on a new write channel model, and we determine and compare bounds on several relevant information-theoretic limits: capacity, symmetric information rate (SIR), information rate with first-order Markov inputs (Markov-$1$ rate), and zero-error capacity.

The remainder of the paper is organized as follows. In Section \ref{sec_bpmr}, we start with a brief description of the write process in bit-patterned media recording, highlighting the main factors leading to write errors. We review some useful mathematical notation in Section \ref{ssec_not}, and then present the new probabilistic, data-dependent write channel model in Section \ref{ssec_model}. We discuss the  types of write errors that can occur -- including substitution-like errors and insertion-deletion errors -- and explain their connection to the channel state process that underlies the model. Two classes of channel state processes are introduced: a Bernoulli state process (Section \ref{ssec_suberr}) and a  binary Markov state process (Section \ref{ssec_iderr}). In Section \ref{ssec_tdmr} we make the observation that, although proposed as a model for BPM recording, the Markov state  channel model is also relevant to high-density magnetic recording using conventional granular media. 
In Section \ref{sec_suberr} we give bounds on the capacity, the SIR, and the Markov-$1$ rate of the Bernoulli state channel. Section \ref{sec_insdel} gives similar bounds on the capacity for the binary Markov state channel. The SIR is numerically computed for both of the channels considered in Section \ref{sec_chnmod}. In Section \ref{sec_kid}, we  introduce a generalization of  binary Markov state channels, namely the $K$-ary Markov state channels. For one such generalized channel, we numerically  estimate the SIR and derive bounds on the channel capacity. Finally, in Section \ref{sec_zec}, we explore the zero-error capacity of the proposed class of write channel models.   Section \ref{sec_conc} provides concluding remarks.

\section{Bit-Patterned Media Recording} \label{sec_bpmr}
A conventional magnetic recording medium is a continuous film of magnetic \emp{grains} that coats the surface of the disk substrate. Each grain is an atomic magnetic unit that assumes one of two possible magnetic states. A group of grains together form a \emp{bit-cell}, an entity that stores one bit of information. Therefore, as the areal information density is increased, the number of grains forming a bit-cell reduces. One of the problems with high-density magnetic recording in conventional media is the \emp{super-paramagnetic effect}, wherein the magnetic states of individual grains change due to the influence of neighbouring grains or due to changes in temperature. When the areal information density is increased to point where there are only a few grains per bit, such uncontrolled changes in the magnetic states of grains are detrimental to reliable information storage.

Bit-patterned media recording (BPMR) proposes to get around this problem by making use of patterned magnetic islands separated by non-magnetic material \cite{whi_97_tmag_bpm}. However, this new structure of the magnetic medium introduces technical challenges not seen in recording on conventional media. An immediate requirement of this media structure is near-perfect synchronization of the write process to ensure that the write head is positioned correctly over the magnetic islands on this disk, i.e., to ensure that the head is positioned within the so-called \emp{writing window zone} of the islands \cite{liv_09_tmag_wwz}. Assuming that this write synchronization is achieved through the use of timing synchronization patterns, there is a possiblity of frequency and/or phase mismatch leading to incorrect writing. Furthermore, even without a timing mismatch, imperfections in the configuration of the patterned magnetic islands may cause writing errors. Finally, as in conventional magnetic recording,  the \emp{switching field distribution} of magnetic grains may contribute to errors in the write process. We will refer to write errors induced by any of these mechanisms as \emp{written-in errors}.

Another important feature of BPMR is the geometry of the writing process. Along the \emp{down-track} direction, the span of influence of the magnetic write field is typically larger than the spacing between the islands. Therefore, at any given time, the write head influences multiple adjacent islands. So, in the process of recording a bit on a specified island, the bit value is also recorded  on a certain number of subsequent islands, with these islands themselves being overwritten in the future by subsequent bits. Fig.~\ref{fig_bpm} gives an illustration of this idealized write process.
\begin{figure}[!ht]
\centering
\includegraphics[width=.95\linewidth]{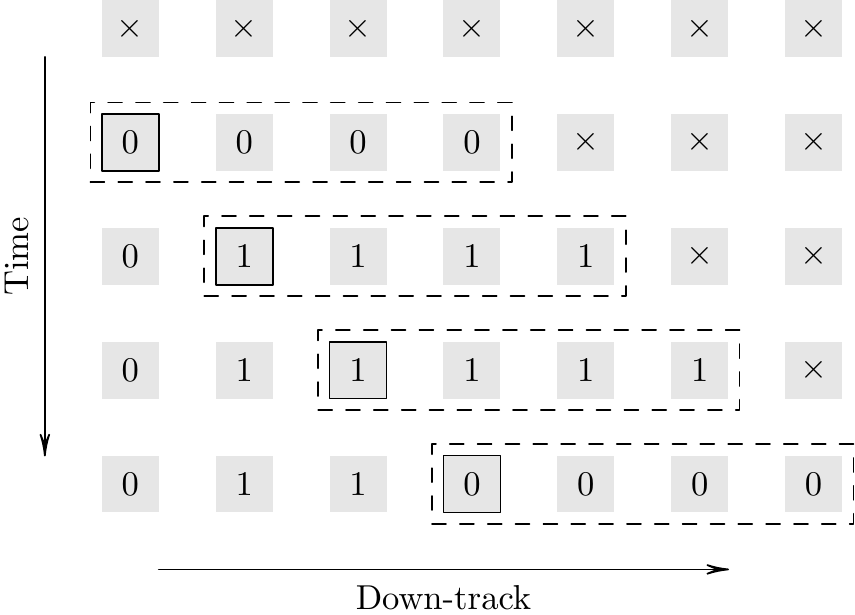}
\caption{Example illustrating writing one bit on each island in BPMR with a write head of writing span $D = 3$. Depicted here are snapshots of the top view of the disk and write mechanism. Initially, all islands, represented here using gray squares, have unknown magnetic states represented by $\times$, which could be $0$ or $1$. At the first time instant, the write head, shown here using a dashed rectangle, is positioned over the first four islands. In this position, the first island is written with the data bit $0$. The three subsequent islands are also written with the first data bit $0$, before being overwritten with their own data. After each bit is written, the head is moved over to the next island along the down-track direction. At each position of the write head, the corresponding ``current'' island is highlighted with a black square boundary.}
\label{fig_bpm}
\end{figure}
We refer to the number of subsequent islands influenced by the write head as the \emph{writing span}, $D$. We will assume throughout this paper that $D \ge 1$, which implies that write process has memory and, as a consequence, the write channel is data-dependent. 

\section{Write Channel Model} \label{sec_chnmod}
\subsection{Notation} \label{ssec_not}
We will denote random variables using capital letters $X$, $Y$, $Z$; random vectors as $X_i^j = (X_i, X_{i + 1}, \cdots, X_j), i, j \in \mbb{N}$ when $i \le j$, a null variable when $i > j$ or when $j$ is non-positive. Infinite-dimensional random vectors are considered to be discrete-time random processes and are denoted using calligraphic letters $\mcl{X}$, $\mcl{Y}$, $\mcl{Z}$. The alphabets over which random processes are defined are denoted by $\mds{X}$, $\mds{Y}$, $\mds{Z}$. Unless otherwise stated, we will denote the channel input variable, process and alphabet using the letters $X$, $\mcl{X}$ and $\mds{X}$ respectively. Similarly, $Y$, $\mcl{Y}$, $\mds{Y}$ will correspond to the channel output, and $Z$, $\mcl{Z}$, $\mds{Z}$ to the channel state ($\mds{Z}$ is not to be confused with the set of integers $\mbb{Z}$). For a real number $a \in [0, 1]$, we define $\overline{a} = 1 - a$.

\subsection{Model} \label{ssec_model}
We assume that the input, output and the channel state alphabet is the binary set $\mds{X} = \mds{Y} = \mds{Z} = \{0, 1\}$. The write channel model can be written as 
\begin{equation}
\label{eq_chnmod}
Y_i = X_i \oplus \left( X_i \oplus X_{i-1}\right) \otimes Z_i,
\end{equation}
where the state process $\mcl{Z}$ is independent of the inputs and the outputs, and $\oplus$ and $\otimes$ represent addition and multiplication in $GF(2)$, respectively. The channel state $Z_i$ is not to be confused with the magnetic state of the island.

It is natural to assume that the channel input and output alphabets are binary since the (intended and actual) magnetic states of the islands can be one of two possible states. The channel state $Z_i$ is a random variable that represents a failure in writing  due to one of the conditions mentioned in the previous section, i.e., when $Z_i = 1$ the write head can be assumed to have failed in writing the intended bit. However, this does not necessarily imply a written-in error because there would be no error if the bit to be written is the same as the existing magnetic state of the island. This is captured by the term $(X_i \oplus X_{i - 1})$ that is multiplied with $Z_i$ in \eqref{eq_chnmod}. This ``noise'' term $(X_i \oplus X_{i - 1}) \otimes Z_i$ is justified because we assume that the timing mismatch or the irregularity of island patterns can cause the write head to be positioned, in the worst case, on the island immediately following the correct island. We will continue with this assumption until Section \ref{sec_kid}, where we construct a more generial write channel model. Also note that when the first bit is written late, i.e. when $Z_1 = 1$, we have $Y_1 = X_0$ which represents the pre-existing magnetic state of the first island. We will assume that $X_0$ is equally likely to be a $0$ or a $1$. Similarly, when the last bit is written late, the last island has a bit in error if the last and the penultimate bits are different. In either case, $n$ magnetic islands are read and their contents are interpreted as the $n$ data bits so that there is no blocklength mismatch.

Based on the $Z_i$ sequence, the channel in \eqref{eq_chnmod} appears to produce different types of errors. When the $Z_i$ sequence consists of isolated ones, the channel appears to produce \emp{substitution} errors. This is illustrated in the Fig.~\ref{fig_subeg}.
\begin{figure}[!ht]
\centering
\includegraphics[width=.95\linewidth]{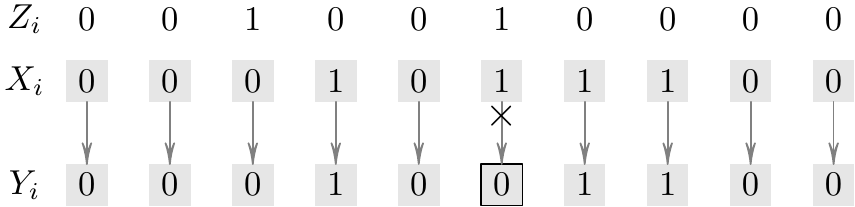}
%\vspace{0.5cm}
\caption{An example of substitution-like errors produced by the write channel. The sequence $Z_i$ gives the channel state for each magnetic island. In this example, $Z_i = 1$ for the third and the sixth islands. The sequence $X_i$ gives the intended magnetization of the islands, i.e., the data to be written. Taking into account the channel states $Z_i$, the sequence $Y_i$ shows the resulting island magnetizations. Note that substitution-like errors occur only when $Z_i = 1$, as was the case with the sixth island here, highlighted with a box around the island in the $Y_i$ sequence. However, not every $Z_i = 1$ results in an error, as illustrated by the third island.}
\label{fig_subeg}
\end{figure}
Such substitution-like errors can occur when the write head misses islands at random and independently of its success in writing on previous islands. Noting that when $Z_i = 0$, $Y_i = X_i$ so that the output reproduces the input exactly, and when $Z_i = 1$, $Y_i = X_{i-1}$ so that the output reproduces the input with a delay of one time instant, we can see that another way to write the relation in \eqref{eq_chnmod} is
\begin{equation} \label{eq_altchnmod}
Y_i = X_{i - Z_i}.
\end{equation}
Now, when the $Z_i$ sequence consists of long runs of ones, the channel appears to produce paired \emp{insertion-deletion} errors, with insertions accompanying $0 \rightarrow 1$ channel state transitions and deletions accompanying $1 \rightarrow 0$ channel state transitions, as shown in Fig.~\ref{fig_ideg}.
\begin{figure}[!ht]
\centering
\includegraphics[width=.95\linewidth]{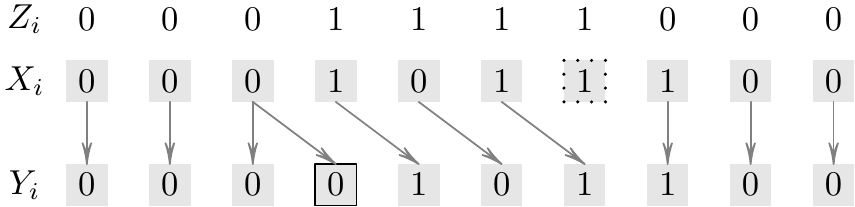}
%\vspace{0.5cm}
\caption{An example of an insertion-deletion error. Due to the channel state sequence $Z_i$, the channel inputs are transformed as shown by the arrows between the $X_i$ and $Y_i$ sequences with the resulting magnetic states of the islands shown in the sequence $Y_i$. The inserted bit in the $Y_i$ sequence accompanying the channel state transition $0\rightarrow 1$ is shown with a box around the corresponding island. The deleted bit from the $X_i$ sequence accompanying the channel state transition $1\rightarrow 0$ is shown with a dotted box around the island.}
\label{fig_ideg}
\end{figure}
These insertion-deletion errors can happen as a result of timing synchronization errors, wherein there is a frequency mismatch between the islands and the write head; or as a result of a group of islands being separated farther than usual or having larger switching fields.

The channel in \eqref{eq_chnmod} and \eqref{eq_altchnmod} is completely defined by specifying the statistics of the $\mcl{Z}$ process. We will consider the channel under two different statistical assumptions on the $\mcl{Z}$ process in the following, and show how this difference in statistics changes the typical behaviour of the channel.

\subsubsection{Bernoulli state channel} \label{ssec_suberr}
When the channel state process $\mcl{Z}$ is an i.i.d. Bernoulli$(\p)$ process, i.e., $\mcl{Z} \sim \mcl{B}(\p)$, we will call the channel the \emp{Bernoulli state channel}. In this case, the channel is completely specified by the parameter $\p$. We shall henceforth denote this channel by $W(\p)$. Typically, for small values of the parameter $\p$ we can expect this channel to produce errors resembling substitution errors (See Fig.~\ref{fig_subeg}).

\subsubsection{Binary Markov state channel} \label{ssec_iderr}
When the channel state process $\mcl{Z}$ is a first-order binary Markov process\footnote{The superscript in the notation $\mcl{M}^{(2)}_1(\p_\msi, \p_\msd)$ denotes the alphabet size over which the process is defined and the subscript denotes the memory in the process. The two arguments give the transition probabilities between the two states.}% In general, when the state alphabet is $\{0, 1, \cdots, K - 1\}$, the first of the two parameters will denote the probability of transition between the states $i \rightarrow (i + 1), i \in \{0, 1, \cdots, K - 2\}$ and the second that of $i \rightarrow (i - 1), i \in \{1, 2, \cdots, K - 1\}$.}
 $\mcl{Z} \sim \mcl{M}_1^{(2)}(\p_{\msi}, \p_{\msd})$, i.e., $\mbb{P}\{Z_i | Z_1^{i - 1}\} = \mbb{P}\{Z_i | Z_{i - 1}\}$, $\mbb{P}\{Z_i = 1 | Z_{i-1} = 0\} = \p_\msi$ and $\mbb{P}\{Z_i = 0 | Z_{i - 1} = 1\} = \p_\msd$, we will refer to the channel as the \emp{binary Markov state channel} and denote it by $W(\p_\msi, \p_\msd)$. As noted earlier, we can expect such a channel to typically produce paired insertion-deletion errors (See Fig.~\ref{fig_ideg}). Hence, the parameters $\p_\msi$ and $\p_\msd$ can be thought of as insertion and deletion probabilities, respectively, of the channel.

We note here that this channel differs from the most general insertion-deletion channel in two ways. First, the inserted bit is always the same as the last written bit. Whereas this seems to be a serious limitation of the model in comparison with the general insertion channel model, it is to be noted that in most practical insertion-deletion channels---channels with synchronization errors, or sticky channels \cite{mit_08_tit_sticky}---the inserted bit is usually the one last written (or transmitted), rather than being a random bit. Second, the insertions and deletions are paired so that one never sees two consecutive insertions or deletions. While this seems to limit the scope of the channel, a straightforward extension to channels with a finite maximum number of consecutive insertions or deletions is possible. However, the estimation of the channel characteristics becomes more complex as the maximum number of consecutive insertions increases (See Section \ref{sec_kid}). Nevertheless, the binary Markov state channel gives us good insight into more general insertion-deletion channels. Moreover, unlike much of the previous work in the literature wherein channels cause only insertion \cite{mit_08_tit_sticky} errors or only deletion \cite{dig_01_all_delchn,dig_06_tit_finbufchn,dri_06_tit_lbcapdc,mit_06_tit_lbcapdc,dig_07_isit_capubdc, fer_10_tit_bdccap} errors, this model considers both insertions and deletions in the same setting. Among works that handle both insertions and deletions, the set up in this paper resembles the duplication-deletion channel of \cite{kir_10_tit_ddccap} more than it does the generic insertion-deletion channel of \cite{hu_10_tcom_insdel}. However, the channel model considered here differs from those in the aforementioned papers in that the insertion-deletion process has memory, i.e. the insertions and deletions are not i.i.d.

\subsection{High-density recording with granular media} \label{ssec_tdmr}
The channel model in \eqref{eq_chnmod} can also be used to describe high-density magnetic recording on conventional granular magnetic media\cite{woo_09_tmag_hdmr}. Although media granularity is typically considered as a two-dimensional phenomenon, we  consider granularity  only in one-dimension -- along the down-track direction,  and assume adjacent tracks to be independent. This simplification allows us to establish lower bounds for performance over the two-dimensional channel as proposed in \cite{woo_09_tmag_hdmr}.

Media granularity in one-dimension results in written-in errors as follows. At storage densities of the order of $1$ bit per grain, the variation in grain size plays an important role in deciding the reliability of data storage. In this case, bit-cells are at most as large as individual grains. We assume that the grains are all $1$ or $2$ bit-cells in size, and that the magnetic state of each grain is decided by the last bit written on them. This is depicted in Fig.~\ref{fig_tdmr}.
\begin{figure}[!ht]
\centering
\includegraphics[width=.95\linewidth]{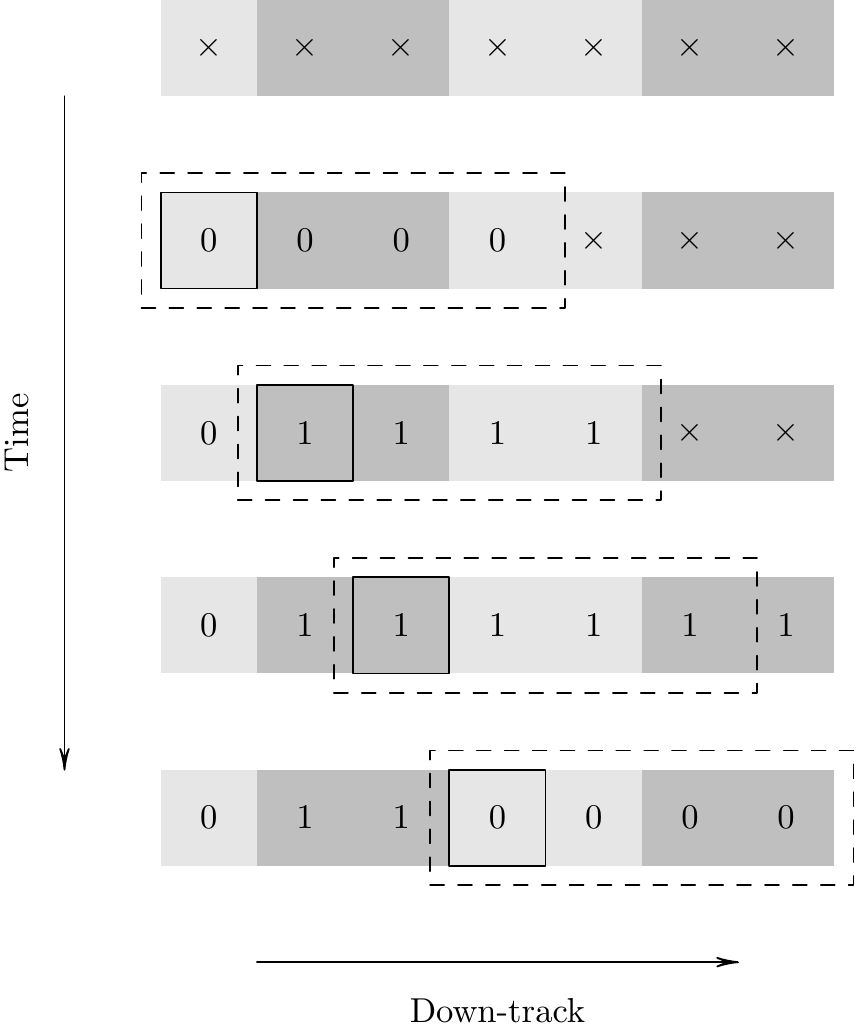}
%\vspace{0.1cm}
\caption{Writing on granular media. Individual bit-cells are shown as $\times$'s in the first row, corresponding to the magnetic states read from these cells, which could be $0$'s or $1$'s. Also shown in shades of gray are the magnetic grains: grains that comprise one bit-cell are shown in light gray and those that comprise two bit-cells are in a darker shade of gray. The two bit-cells comprising a grain of size $2$ always have the same magnetic state. As in the case of BPMR, the write head spans multiple bit-cells, as shown by the dashed rectangles. The ``current'' bit-cell at each time instant is shown with a black square.}
\label{fig_tdmr}
\end{figure}
Grains that are $2$ bit-cells large will result in written-in errors if the two bits written on them are different. Fig.~\ref{fig_tdmeg} gives an illustration of the written-in errors in this scenario.
\begin{figure}[!ht]
\centering
\includegraphics[width=.95\linewidth]{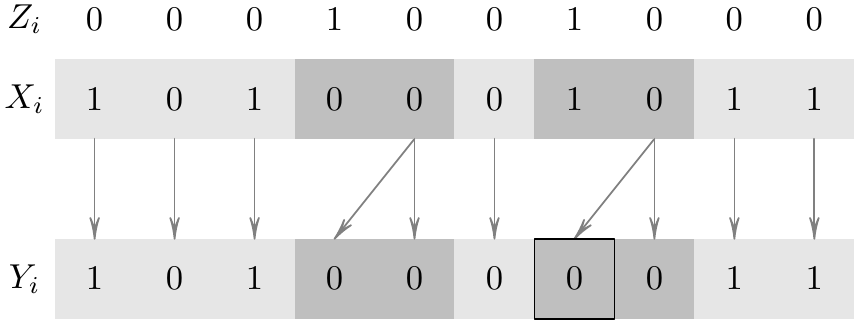}
%\vspace{0.1cm}
\caption{Written-in errors due to media granularity. The $\mcl{Z}$ process represents the grain pattern with $Z_i$ being $1$ when the corresponding bit-cell is the first bit-cell in a grain that comprises two bit-cells. The $X_i$ sequence shows the data to be written. The arrows between the $X_i$ and $Y_i$ sequences show the transformation of information according to the grain pattern, resulting in the sequence $Y_i$ being stored. The written-in errors are shown with a box around the corresponding bit-cell.}
\label{fig_tdmeg}
\end{figure}
Comparing Figures \ref{fig_ideg} and \ref{fig_tdmeg}, it is easy to see that the channel model in this case can be written as
\[
Y_i = X_i \oplus (X_i \oplus X_{i + 1}) \otimes Z_i.
\]
Using a simple time-reversal argument, it can be seen that this is exactly the channel in \eqref{eq_chnmod}. However, in this case, the $Z_i$ sequence cannot have two consecutive ones, i.e., it satisfies the $(1, \infty)$ run-length constraint \cite[Chap. 4]{imm_99_bok_codsto}. We choose to model this as the $W(\p_\msi, 1)$ channel since any realization of the $\mcl{M}_1^{(2)}(\p_\msi, 1)$ process meets the run-length requirement. In \cite{maz_10_isit_hdmr}, the authors consider this channel and derive upper bounds for the achievable rates over the channel. Our focus here will be on obtaining bounds on the achievable information rates for the general channel in  \eqref{eq_chnmod} in the context of the two channel state processes defined in Section \ref{ssec_model}.

\section{Bernoulli State Channel} \label{sec_suberr}
The channel space for the Bernoulli state channel defined in the previous section is parameterized by $\p \in [0,1]$. Since the channel has memory, its capacity is given as
\begin{align}
C &= \lim_{n\tends\infty}\sup_{\mbb{P}\{X_1^n\}}\frac{1}{n}I\left(X_1^n;Y_1^n\right). \notag
\end{align}
All the entropy and mutual information terms calculated are in bits, so that the capacity is always measured in bits per channel use. We will denote by $C(\p)$ the capacity of the Bernoulli state channel with parameter $\p$.

\begin{prop}[Channel symmetry] \label{prop_chnsym}
$C(\p) = C(\overline{\p})$.
\end{prop}
\begin{proof}
Let us denote by $I_{\p}(X_1^n; Y_1^n)$ the mutual information between the vectors $X_1^n$ and $Y_1^n$ when the channel parameter is $\p$. Note that the channel output can be simultaneously written as
$$
Y_i = ( X_i \otimes \overline{Z}_i ) \oplus ( X_{i-1} \otimes Z_i)
$$  
\noindent
and
$$
Y_i = ( S_i \otimes Z_i ) \oplus ( S_{i + 1} \otimes \overline{Z}_i )
$$
where $S_i = X_{i - 1}$ and $\overline{Z}_i = Z_i \oplus 1$. Thus
\[
I_{\p}(X_i; Y_1^n | X_1^{i - 1}) = I_{\overline{\p}}(S_{n - i + 1}; Y_1^n | S_{n - i + 2}^n) 
\]
for all but a vanishing fraction of indices $i \in \mbb{N}$, as $n\tends\infty$.  
Therefore,
\begin{align}
C(\p) &= \lim_{n\tends\infty}\sup_{\mbb{P}\{X_1^n\}}\frac{1}{n}I_{\p}(X_1^n; Y_1^n) \notag
\end{align}
\begin{align}
&= \lim_{n\tends\infty}\sup_{\mbb{P}\{X_1^n\}}\frac{1}{n}\sum_{i=1}^nI_{\p}(X_i; Y_1^n | X_1^{i - 1}) \notag \\
&= \lim_{n\tends\infty}\sup_{\mbb{P}\{X_1^n\}}\frac{1}{n}\sum_{i=1}^nI_{\overline{\p}}(S_{n - i + 1}; Y_1^n | S_{n - i + 2}^n) \notag \\
&= \lim_{n\tends\infty}\sup_{\mbb{P}\{X_1^n\}}\frac{1}{n}I_{\overline{\p}}(S_1^n; Y_1^n) \notag \\
&= \lim_{n\tends\infty}\sup_{\mbb{P}\{X_1^n\}}\frac{1}{n}I_{\overline{\p}}(X_1^{n - 1}; Y_1^n) \notag \\
&= \lim_{n\tends\infty}\sup_{\mbb{P}\{X_1^n\}}\frac{1}{n}I_{\overline{\p}}(X_1^n; Y_1^n) = C(\overline{\p}). \notag
\end{align}
The channel space can therefore be reduced to the interval $\p\in[0,\frac{1}{2}]$. Further note that the same symmetry argument holds for not just the rate-maximizing input distribution, but for all input distributions.
\end{proof}

The capacity of the Bernoulli state channel is upper bounded by the achievable rate for a genie-aided decoder, i.e., one with the $\mcl{Z}$ process realization known. Given the realization of the $\mcl{Z}$ process, the inserted bits and the positions of the deleted bits are known so that the Bernoulli state channel is equivalent to a correlated erasure channel with average erasure rate $\mbb{P}\{Z_{i - 1} = 1, Z_i = 0\} = \p\overline{\p}$. The resulting erasure channel is a correlated channel since, erasures being dependent on $1 \rightarrow 0$ transitions in the $\mcl{Z}$ process, two consecutive bits cannot be erased. Therefore,
\begin{equation} \label{eq_capuber}
%C(\p) \le 1 - \p\overline{\p} = C_{g\e}(\p)
C(\p) \le 1 - \p\overline{\p} \stackrel{\rm \Delta}{=} C_{g\e}(\p)
\end{equation}
since the capacity of a correlated erasure channel is the same as that of a memoryless erasure channel with the same erasure probability. We call this upper bound, $C_{g\e}(\p)$, the \emp{genie-erasure capacity} of the channel $W(\p)$.

Consider 
\begin{align}
\frac{1}{n}I\left(X_1^n; Y_1^n\right) &= \frac{1}{n}H(Y_1^n) - \frac{1}{n}H(Y_1^n|X_1^n) \notag \\
&\stackrel{(a)}{=} \frac{1}{n}H(Y_1^n) - \frac{1}{n}\sum_{i=1}^nH(Y_i|X_{i-1}^{i}) \notag \\
&= \frac{1}{n}H(Y_1^n) - \frac{h_2(\p)}{n}\sum_{i=1}^n\mbb{P}\{X_i\ne X_{i-1}\}, \label{eq_mi}
\end{align}
where $h_2(\cdot)$ is the \emp{binary entropy function} \cite{cov_06_bok_infoth}. The equality labelled $(a)$ follows from the definition of $Y_i$ in \eqref{eq_chnmod} and the fact that $\mcl{Z} \sim \mcl{B}(\p)$. Since the 
information-rate-maximizing input distribution is unknown, we will now derive lower bounds to the capacity by making certain assumptions about the statistics of the input process $\mcl{X}$.

\subsection{I.i.d. input process} \label{sec_suberriid}
We will first assume that the input process $\mcl{X}$ is an i.i.d. $\mcl{B}(\alpha)$ process. With this assumption, the maximum achievable information rate, called the \emp{i.i.d.~capacity}, denoted $C_{iid}(\alpha, \p)$, gives a lower bound to the capacity.
\begin{prop}[Input symmetry] \label{prop_iidipsym}
$C_{iid}(\alpha, \p) = C_{iid}(\overline{\alpha}, \p)$. 
\end{prop}
\begin{proof}
Let $X_1^n$ and $Y_1^n$ be the input and output respectively of the channel $W(\p)$. Define $\hat{X}_1^n = (X_1 \oplus 1, X_2 \oplus 1, \cdots, X_n \oplus 1)$ and $\hat{Y}_1^n = (Y_1 \oplus 1, Y_2 \oplus 1, \cdots, Y_n \oplus 1)$. When $X_1^n \sim \mcl{B}(\alpha)$, $\hat{X}_1^n \sim \mcl{B}(\overline{\alpha})$. Further since $X_1^n \leftrightarrow \hat{X}_1^n$ and $Y_1^n \leftrightarrow \hat{Y}_1^n$ are bijections, we have
\[
I_\p(X_1^n; Y_1^n) = I(\hat{X}_1^n; \hat{Y}_1^n).
\]
Clearly, $\hat{X}_1^n$ and $\hat{Y}_1^n$ also satisfy the relation in  \eqref{eq_chnmod} and consequently
\[
I(\hat{X}_1^n; \hat{Y}_1^n) = I_\p(\hat{X}_1^n; \hat{Y}_1^n).
\]
From this and the fact that
\[
C_{iid}(\alpha, \p) = \lim_{n \tends \infty}\frac{1}{n}I_\p(X_1^n; Y_1^n) \Big|_{\mcl{X} \sim \mcl{B}(\alpha)}
\]
we have the desired result.
\end{proof}
As a consequence of Propositions \ref{prop_chnsym} and \ref{prop_iidipsym}, we have
\[
C_{iid}(\alpha, \p) = C_{iid}(\alpha, \overline{\p}) = C_{iid}(\overline{\alpha}, \overline{\p}) = C_{iid}(\overline{\alpha}, \p).
\]

From Proposition \ref{prop_iidipsym} and the fact that $C_{iid}(\alpha, \p)$ is concave in $\alpha$, we immediately have the following.
\begin{cor}[Rate-maximizing i.i.d. distribution] \label{cor_sir}
\[\max_{\alpha \in [0, 1]}C_{iid}(\alpha, \p) = C_{iid}(\frac{1}{2}, \p){\ }\forall{\ }\p \in [0, 1].\]
\end{cor}
When the binary input is i.i.d. with a uniform distribution (sometimes abbreviated as $i.u.d.$), the corresponding rate is called the \emp{symmetric information rate} (SIR) and is denoted $C_{iud}(\p)$. The SIR is of interest because it can be achieved by a random linear coset code \cite{gal_68_bok_infoth}.

We have from \eqref{eq_mi}
\begin{align}
C(\p) \ge C_{iud}(\p) &= \mcl{H}(\mcl{Y})\Big|_{\mcl{X} \sim \mcl{B}(1/2)} - \frac{h_2(\p)}{2} \label{eq_iudcap}
\end{align}

We can lower bound the SIR by disregarding the data-dependence of the noise in the channel. This gives a channel equivalent to a BSC with crossover probability $\p / 2$ so that 
\begin{equation} \label{eq_iudlb0}
C_{iud}(\p) \ge 1 - h_2\left(\frac{\p}{2}\right) \stackrel{\Delta}{=} L_0^{iud}(\p).
\end{equation}
Further lower bounds can be obtained by conditioning the entropy of the output as follows:
\begin{align}
C_{iud}(\p) &= \lim_{n\tends\infty}\frac{1}{n}\sum_{i=1}^nH(Y_i|Y_1^{i-1}) - \frac{h_2(\p)}{2} \notag \\
&\ge \lim_{n\tends\infty}\frac{1}{n}\sum_{i=1}^nH(Y_i|Y_1^{i-1}, X_{i-1}) - \frac{h_2(\p)}{2} \notag \\
&= h_2\left(\frac{1 - \p}{2}\right) - \frac{h_2(\p)}{2} \stackrel{\Delta}{=} L_1^{iud}(\p), \label{eq_iudlb1}
\end{align}
where we have used the fact that $Y_i$ depends only on $X_{i-1}$ and $X_i$, and given $X_{i-1}$, $Y_i$ is independent of $Y_1^{i-1}$. Continuing as above, we can obtain a tighter lower bound
\begin{align}
C_{iud}(\p) &\ge \lim_{n\tends\infty}\frac{1}{n}\sum_{i=1}^nH(Y_i|Y_1^{i-1},X_{i-2}) - \frac{h_2(\p)}{2} \notag \\
&= \left(\frac{1 + \p}{2}\right)h_2\left(\frac{1 + \p^2}{2(1 + \p)}\right) \notag \\
&\hspace*{2mm} + \left(\frac{1 - \p}{2}\right)h_2\left(\frac{1 - \p}{2}\right) - \frac{h_2(\p)}{2}  \stackrel{\Delta}{=} L_2^{iud}(\p). \label{eq_iudlb2}
\end{align}
A straightforward upper bound for the SIR, implied by \eqref{eq_iudcap} is  
\begin{equation} \label{eq_iudub0}
C_{iud}(\p) \le 1 - \frac{h_2(\p)}{2} \stackrel{\Delta}{=} U_0^{iud}(\p),
\end{equation} 
which follows because the entropy rate for a binary process $\mcl{H}(\mcl{Y}) \le 1$. Note that this bound is achieved when $\mcl{Y}$ is the i.i.d. $\mcl{B}(1/2)$ process. Again, starting from \eqref{eq_iudcap}, we can obtain upper bounds for the SIR by removing conditioning from the entropy of the output, as shown below:
\begin{align}
C_{iud}(\p) &\le \lim_{n\tends\infty}\frac{1}{n}\sum_{i=1}^nH(Y_i|Y_{i-1}) - \frac{h_2(\p)}{2} \notag \\
&= h_2\left(\frac{1-\p\overline{\p}}{2}\right) - \frac{h_2(\p)}{2} \stackrel{\Delta}{=} U_1^{iud}(\p). \label{eq_iudub1}
\end{align}
As with the lower bounds, we can find a tighter upper bound for the entropy rate $\mcl{H}(\mcl{Y})$ as follows
\begin{align}
C_{iud}(\p) &\le \lim_{n\tends\infty}\frac{1}{n}\sum_{i=1}^nH(Y_i|Y_{i-2}^{i-1}) - \frac{h_2(\p)}{2} \notag \\
&= \frac{1 - \p\overline{\p}}{2}h_2\left(\frac{1}{2(1 - \p\overline{\p})}\right) \notag \\
&\hspace*{5mm}+ \frac{1 + \p\overline{\p}}{2}h_2\left(\frac{1}{2(1 + \p\overline{\p})}\right) - \frac{h_2(\p)}{2} \stackrel{\Delta}{=} U_2^{iud}(\p). \label{eq_iudub2}
\end{align}

Fig.~\ref{fig_capbounds} plots the lower and the upper bounds for SIR discussed above.
\begin{figure}[ht]
\centering
\includegraphics[width=0.95\linewidth]{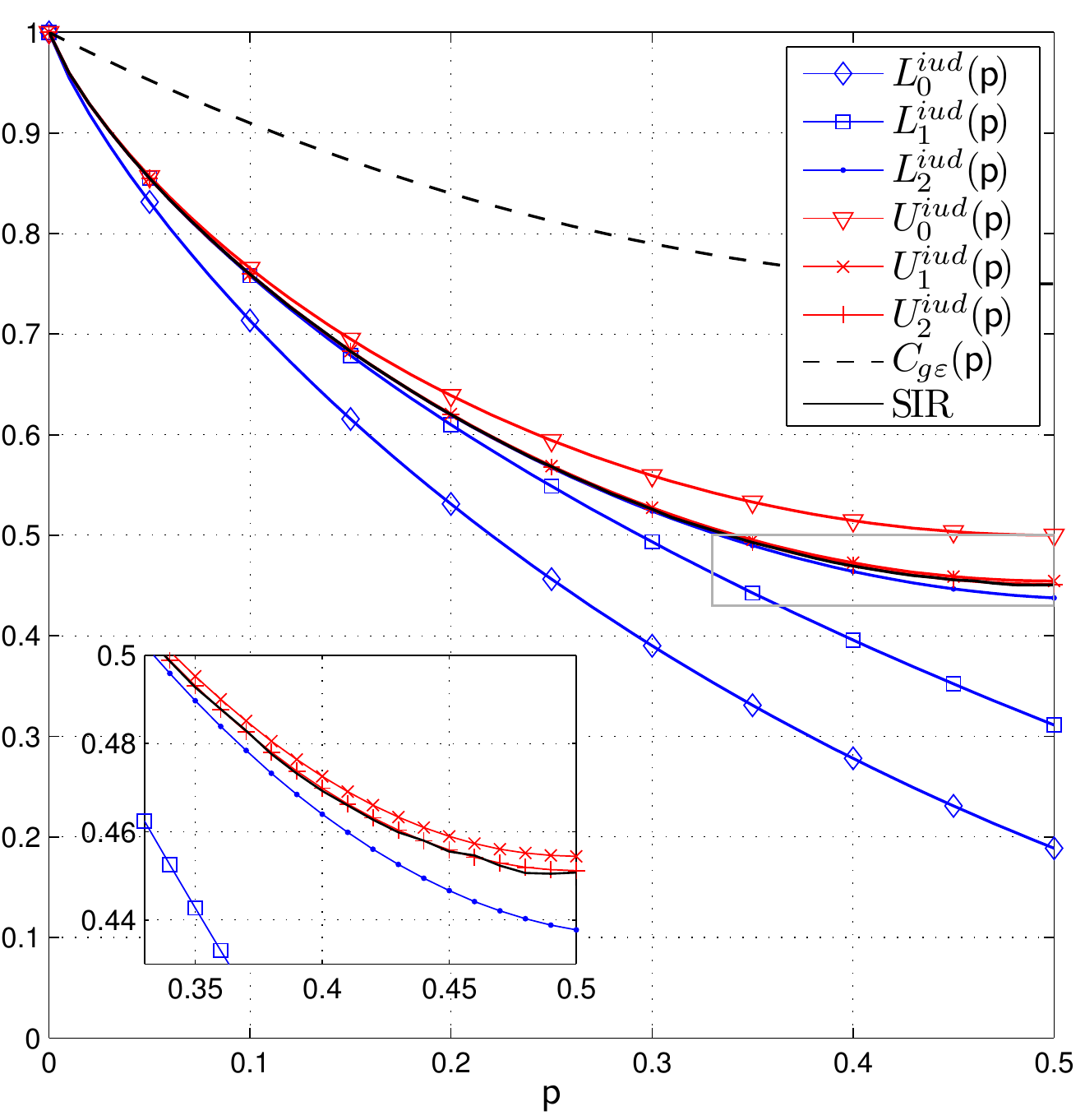}
\caption{Lower and upper bounds for SIR given in \eqref{eq_iudlb0}, \eqref{eq_iudlb1}, \eqref{eq_iudlb2}, \eqref{eq_iudub0}, \eqref{eq_iudub1} and \eqref{eq_iudub2}, and the upper bound for capacity in \eqref{eq_capuber}. The SIR from simulations is also shown.}
\label{fig_capbounds}
\end{figure}
Note that the tighter lower and upper bounds --- $L_2^{iud}(\p)$ in \eqref{eq_iudlb2} and $U_2^{iud}(\p)$ in \eqref{eq_iudub2} --- almost coincide for $\p \le 0.3$ (and from symmetry, for $\p \ge 0.7$). In this range, therefore, where the bounds themselves are greater than 0.5, they approximate the SIR fairly accurately. The process $\mcl{Y}$ is ergodic, as all channel states can be reached within a finite number of steps at any time with strictly positive probability \cite[\S 5.3]{ric_08_bok_mct}, and it converges to a stationary process. As a result, from the Asymptotic Equipartition Property \cite[\S 3.1]{cov_06_bok_infoth}, we have the entropy rate
\[
\mcl{H}(\mcl{Y}) = \lim_{n\tends\infty}\frac{1}{n}H(Y_1^n) = -\lim_{n\tends\infty}\frac{1}{n}\log_2\mbb{P}\{Y_1^n\},
\] 
which can be numerically evaluated through a forward pass of the BCJR algorithm \cite{bah_74_tit_bcjr}, \cite[\S 5.3]{ric_08_bok_mct}. By using long enough sequences $X_1^n$ and $Y_1^n$ in the computation, the SIR can be obtained with an accuracy of $O(\frac{1}{\sqrt{n}})$. This is shown as the ``SIR'' curve in Fig.~\ref{fig_capbounds}, from which we conclude that the upper bound $U_2^{iud}(\p)$ in \eqref{eq_iudub2} is a good approximation for the SIR.

\subsection{First-order Markov Input process} \label{sec_suberrbss}
To explore the loss in the achievable rate due to independent uniformly distributed input, we let the source have memory. We consider a first-order binary Markov input process $\mcl{X}$. Taking a cue from the input symmetry of the Bernoulli state channel, we consider a \emp{symmetric} binary Markov process $\mcl{X}$
%described as $\mcl{X} \in \{0,1\}^*$ 
with $\mbb{P}\{X_i = 1 | X_{i-1} = 0\} = \mbb{P}\{X_i = 0 | X_{i - 1} = 1\} = \beta$. We denote this by $\mcl{X} \sim \mcl{M}_1^{(2)}(\beta)$, meaning that $\mcl{X}$ is a binary (alphabet of size $2$) Markov source with memory $1$ and transition parameter $\beta$.

Starting from  \eqref{eq_mi}, we can arrive at lower and upper bounds for the \emph{Symmetric Markov-$1$ Rate} (M1R),  $C_{M1}(\beta, \p)$, which we define as the maximum rate of information transfer when $\mcl{X} \sim \mcl{M}_1^{(2)}(\beta)$. The lower bounds analogous to those for the SIR are
\begin{align}
C_{M1}(\beta, \p) &\ge 1 - h_2\left(\frac{\p}{2}\right) \stackrel{\Delta}{=} L_0^{M1}(\p), \notag
\end{align}
which is the same as $L_0^{iud}(\p)$ in \eqref{eq_iudlb0},
\begin{align}
C_{M1}(\beta, \p) &= \mcl{H}(\mcl{Y})\Big|_{\mcl{X} \sim \mcl{M}_1^{(2)}(\beta)} - \beta h_2(\p) \notag \\
&\ge \lim_{n\tends\infty}\frac{1}{n}\sum_{i=1}^nH(Y_i|Y_1^{i-1}, X_{i-1}) - \beta h_2(\p) \notag \\
&= h_2(\beta\overline{\p}) - \beta h_2(\p) \stackrel{\Delta}{=} L_1^{M1}(\beta, \p), \notag % \label{eq_bsslb1}
\end{align}
and
\begin{align}
C_{M1}(\beta, \p) &\ge \lim_{n\tends\infty}\frac{1}{n}\sum_{i=1}^nH(Y_i|Y_1^{i-1}, X_{i-2}) - \beta h_2(\p) \notag \\
&= \left(1 - \beta\overline{\p}\right)h_2\left(\frac{\beta\left(1 - \beta(1 - \p^2)\right)}{1 - \beta\overline{\p}}\right) \notag \\
&\hspace*{2mm}+ \beta\overline{\p}h_2\left(\beta\overline{\p}\right) - \beta h_2(\p) \stackrel{\Delta}{=} L_2^{M1}(\beta, \p). \label{eq_bsslb2}
\end{align}
The trivial upper bound analogous to $U_0^{iud}(\p)$ in \eqref{eq_iudub0} is
\begin{align}
C_{M1}(\beta, \p) \le 1 - \beta h_2(\p) \stackrel{\Delta}{=} U_0^{M1}(\beta, \p). \notag
\end{align}
The upper bounds corresponding to $U_1^{iud}$ and $U_2^{iud}$ are
\begin{align}
C_{M1}(\beta, \p) &\le h_2\left(1 - \beta + 2\beta^2\p\overline{\p}\right) - \beta h_2(\p) \stackrel{\Delta}{=} U_1^{M1}(\beta, \p) \notag % \label{eq_bssub1}
\end{align}
and
\begin{align}
C_{M1}(\beta&, \p) \le \left\{\beta^2\bb(1 - \p\overline{\p}) + \beta^2(3 - \beta)\p\overline{\p} + (1 + \beta)\bb^2\right\} \notag\\
&\times h_2\left(\frac{\beta^2\bb(1 - \p\overline{\p}) + \beta^3\p\overline{\p} + \beta\bb^2}{\beta^2\bb(1 - \p\overline{\p}) + \beta^2(3 - \beta)\p\overline{\p} + (1 + \beta)\bb^2}\right) \notag \\
&{\ }{\ }{\ }{\ } + \left\{2\beta^2\bb(1 - \p\overline{\p}) + \beta^3(1 - 2\p\overline{\p}) + \beta\bb^2\right\} \notag\\
&\times h_2\left(\frac{\beta^2\bb(1 - \p\overline{\p}) + \beta^3\p\overline{\p} + \beta\bb^2}{2\beta^2\bb(1 - \p\overline{\p}) + \beta^3(1 - 2\p\overline{\p}) + \beta\bb^2}\right) \notag \\
&{\ }{\ }{\ }{\ } - \beta h_2(\p) \stackrel{\Delta}{=} U_2^{M1}(\beta, \p). \label{eq_bssub2}
\end{align}
Fig.~\ref{fig_bsscap} shows the contours of the bounds for $C_{M1}$ in \eqref{eq_bsslb2} and \eqref{eq_bssub2}. (Only the tighter bounds, $L_2^{M1}(\beta, \p)$ and $U_2^{M1}(\beta, \p)$ are shown.)
\begin{figure}[ht]
\centering
\includegraphics[width=0.95\linewidth]{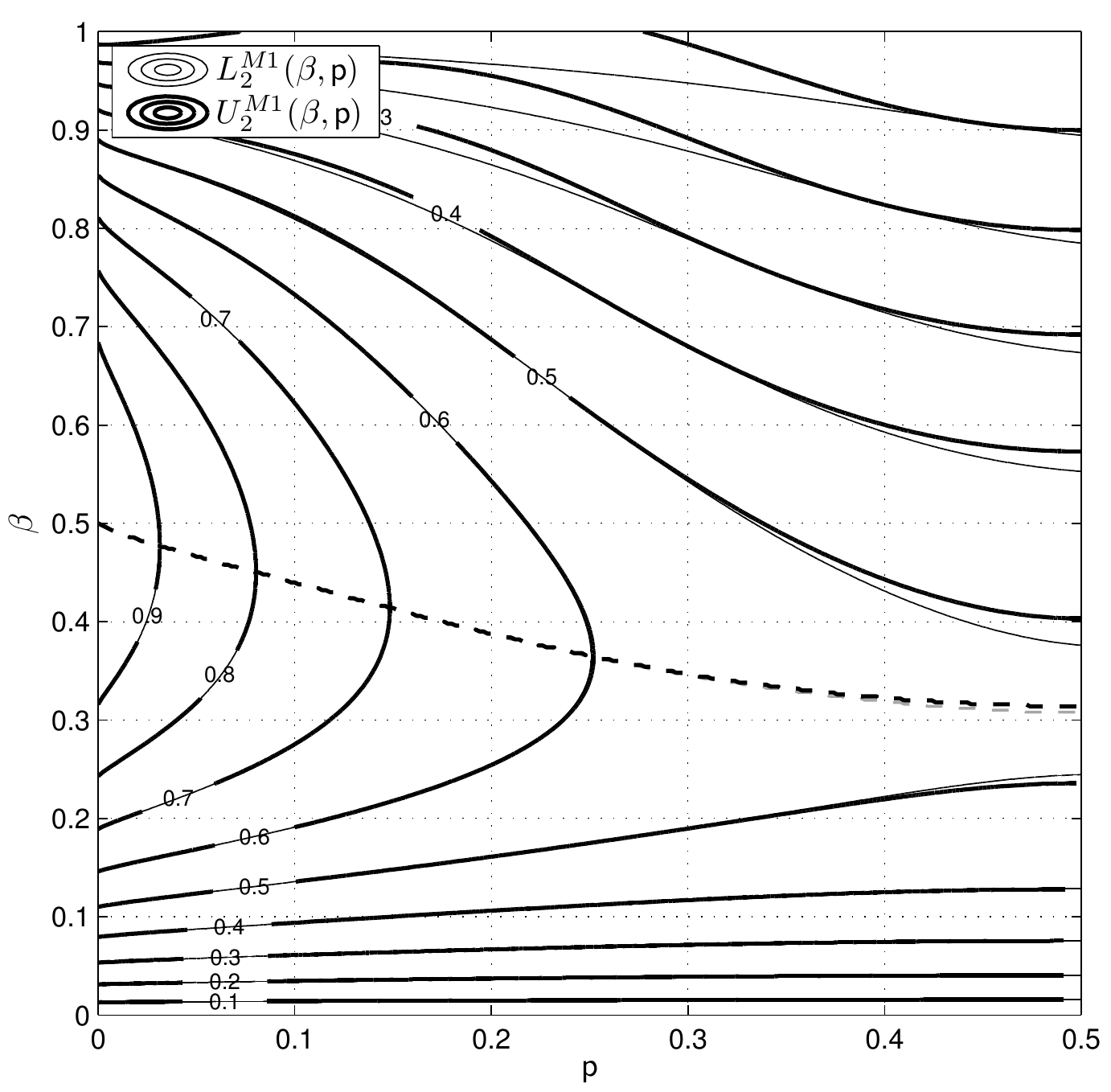}
\caption{Contours of lower (thin curves) and upper (thick curves) bounds for M1R given in \eqref{eq_bsslb2} and \eqref{eq_bssub2}.}
\label{fig_bsscap}
\end{figure}
As was the case for i.u.d. input, the tighter bounds \eqref{eq_bsslb2} and \eqref{eq_bssub2} have almost coinciding contours for a wide range of parameters $(\beta, \p)$. Unlike the case of i.i.d. inputs, the rate-maximizing input parameter $\beta^*(\p)$ is not easily obtained. A close estimate can be obtained by maximizing the bounds obtained above. These are shown (dotted lines) in Fig.~\ref{fig_bsscap}. Since the optimal $\beta$ values for the tighter lower and upper bounds $L_2^{M1}(\beta, \p)$ and $U_2^{M1}(\beta, \p)$ almost coincide, we can say that $\beta^*(\p)$ is monotonically decreasing in $\p$ in the interval $[0,\frac{1}{2}]$ with $\beta^*(0) = \frac{1}{2}$ and $\beta^*(1/2) \approx 0.31$.

Fig.~\ref{fig_iidbsscap} compares the SIR (solid line representing $L_2^{iud}(\p)$ in \eqref{eq_iudlb2}, and dashed line representing $U_2^{iud}(\p)$ in \eqref{eq_iudub2}) and the M1R (solid line for $L_2^{M1}(\beta^*_{L_2}(\p), \p)$ in \eqref{eq_bsslb2}, dashed line for $U_2^{M1}(\beta^*_{U_2}(\p), \p)$ in \eqref{eq_bssub2}) over the range of $\p$ values.
\begin{figure}[!ht]
\centering
\includegraphics[width=0.95\linewidth]{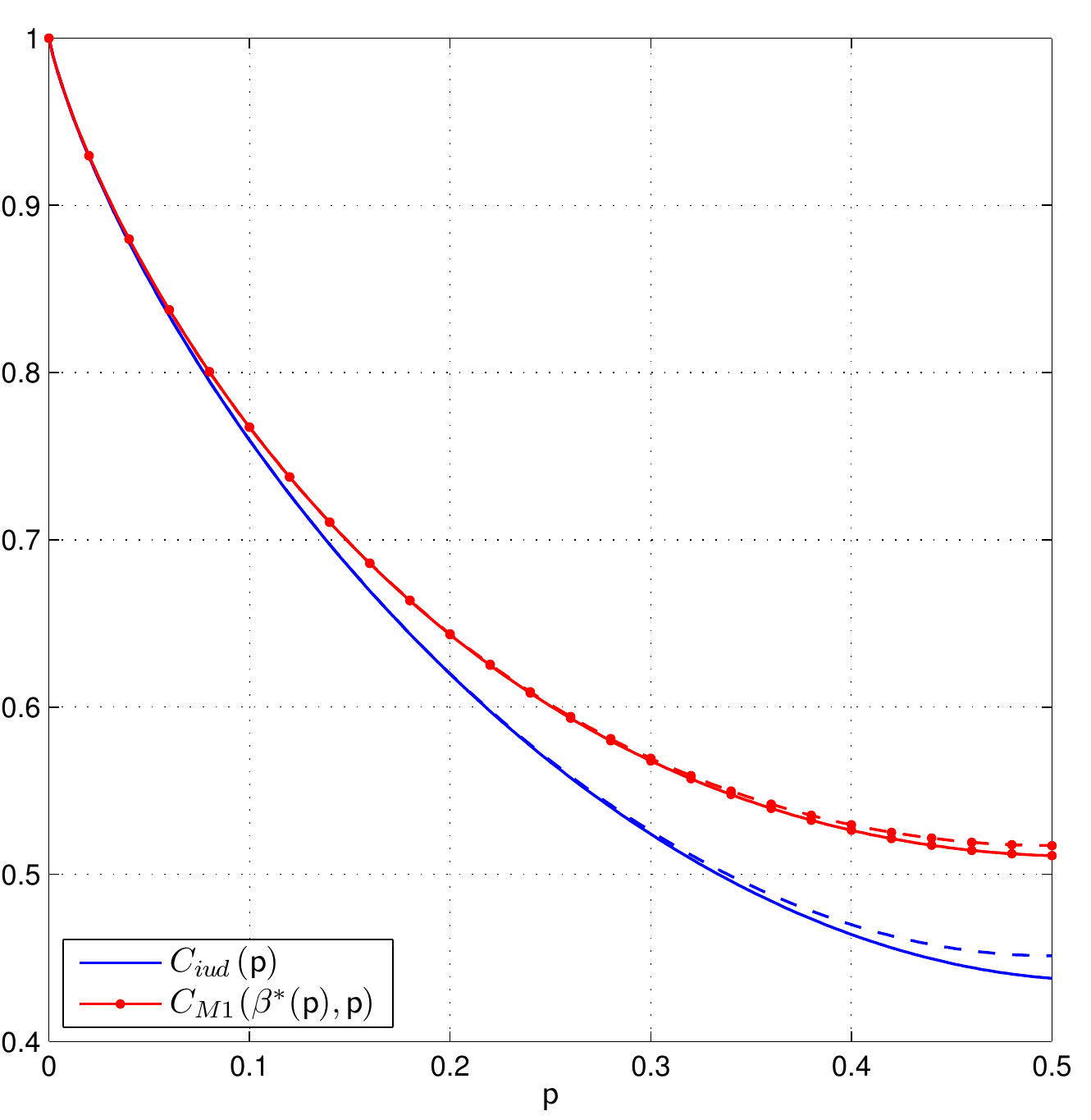}
\caption{Comparison between the tighter lower and upper bounds for SIR and M1R.}
\label{fig_iidbsscap}
\end{figure}
It is clear that considerable gains in reliable information transfer rate are possible by using an input with memory. In particular, note that whereas there is a range of $\p$ values for which the SIR is smaller than $0.5$, the M1R is strictly larger than $0.5{\ }\forall{\ }\p$. It is clear that for a sequence of input Markov processes of increasing orders, the achievable rates are non-decreasing. The algorithm suggested in \cite{kav_01_glb_marcap, von_08_tit_gbaa} can be employed to optimize the input Markov process to maximize the rate. %arn_06_tit_irchnmem, 

\section{Binary Markov State Channel} \label{sec_insdel}
The channel space for the binary Markov state channel defined in Section \ref{ssec_iderr} is $(\p_\msi, \p_\msd) \in [0,1]^2$. Note that the channel space is ordered, i.e., the first parameter is the $0 \rightarrow 1$ transition (insertion) probability and the second the $1 \rightarrow 0$ transition (deletion) probability. As in the case of the Bernoulli state channel, the capacity of the binary Markov state channel, denoted $C(\p_\msi, \p_\msd)$, is given as 
\begin{align}
C(\p_\msi, \p_\msd) = \lim_{n\tends\infty} \sup_{\mbb{P}\{X_1^n\}} \frac{1}{n} I_{(\p_\msi, \p_\msd)}(X_1^n; Y_1^n). \notag
\end{align}
We shall assume that the channel has converged to the stationary distribution. This means that when $\p_\msi$ (resp., $\p_\msd$) is zero, the channel is the noise-free channel (resp., noise-free channel with a delay) and hence $C(0, \p_\msd) = 1$ (resp., $C(\p_\msi, 0) = 1$). We first establish the following symmetry property of the binary Markov state channel.
\begin{prop}[Channel symmetry] \label{prop_1idsym}
$C(\p_\msi, \p_\msd) = C(\p_\msd, \p_\msi)$.
\end{prop}
\begin{proof}
We know that if $Z_1^n = (Z_1, Z_2, \cdots, Z_n)$ is a Markov process, then so is $(Z_n, Z_{n - 1}, \cdots, Z_1)$ \cite[\S 16-4]{pap_91_bok_prob}. Furthermore, since the channel is assumed to have converged to the stationary distribution, the conditional distributions $\mbb{P}\{Z_i | Z_{i - 1}\}$ and $\mbb{P}\{Z_i | Z_{i + 1}\}$ are identical.  However, note that whereas a transition $Z_{i - 1} = 0 \rightarrow Z_i = 1$ is an insertion, the transition $Z_{i + 1} = 0 \rightarrow Z_i = 1$ is a deletion. This implies that $I_{(\p_\msi, \p_\msd)}(X_i; Y_1^n | X_1^{i - 1}) = I_{(\p_\msd, \p_\msi)}(S_{n - i + 1}; Y_1^n | S_{n - i + 2}^n)$, where $S_i = X_{i - 1}$, for all but a vanishing fraction of indices $i$, as $n\tends\infty$. Therefore,
\begin{align}
C(\p_\msi, \p_\msd) &= \lim_{n\tends\infty}\sup_{\mbb{P}\{X_1^n\}}\frac{1}{n}I_{(\p_\msi, \p_\msd)}(X_1^n; Y_1^n) \notag \\
&= \lim_{n\tends\infty}\sup_{\mbb{P}\{X_1^n\}}\frac{1}{n}\sum_{i=1}^nI_{(\p_\msi, \p_\msd)}(X_i; Y_1^n | X_1^{i - 1}) \notag \\
&= \lim_{n\tends\infty}\sup_{\mbb{P}\{X_1^n\}}\frac{1}{n}\sum_{i=1}^nI_{(\p_\msd, \p_\msi)}(S_{n - i + 1}; Y_1^n | S_{n - i + 2}^n) \notag \\
&= \lim_{n\tends\infty}\sup_{\mbb{P}\{X_1^n\}}\frac{1}{n}I_{(\p_\msd, \p_\msi)}(S_1^n; Y_1^n) \notag \\
&= \lim_{n\tends\infty}\sup_{\mbb{P}\{X_1^n\}}\frac{1}{n}I_{(\p_\msd, \p_\msi)}(X_1^{n - 1}; Y_1^n) \notag \\
&= \lim_{n\tends\infty}\sup_{\mbb{P}\{X_1^n\}}\frac{1}{n}I_{(\p_\msd, \p_\msi)}(X_1^n; Y_1^n) = C(\p_\msd, \p_\msi). \notag
\end{align}
The channel space can therefore be reduced to the region $\p_\msi \in [0, 1], \p_\msd \in [0, \p_\msi]$. As was the case for the $W(\p)$ channel, we have the same symmetry for any fixed input distribution.
\end{proof}
As a consequence of Proposition~\ref{prop_1idsym}, we can assume an unordered pair $\{\p_\msi, \p_\msd\}$ parameterizing the channel space.

As for the Bernoulli state channel, we can define the genie-erasure capacity $C_{g\e}(\p_\msi, \p_\msd)$ of the binary Markov state channel $W(\p_\msi, \p_\msd)$. In this case, the resulting channel is a correlated erasure channel with an average erasure probability $\mbb{P}\{Z_{i - 1} = 1, Z_i = 0\} = \frac{\p_\msi\p_\msd}{\p_\msi + \p_\msd}$, so that
\begin{equation} \label{eq_1idcapub}
C(\p_\msi, \p_\msd) \le 1 - \frac{\p_\msi\p_\msd}{\p_\msi + \p_\msd} \stackrel{\Delta}{=} C_{g\e}(\p_\msi, \p_\msd).
\end{equation}

\begin{figure}[!ht]
\centering
\includegraphics[width=.95\linewidth]{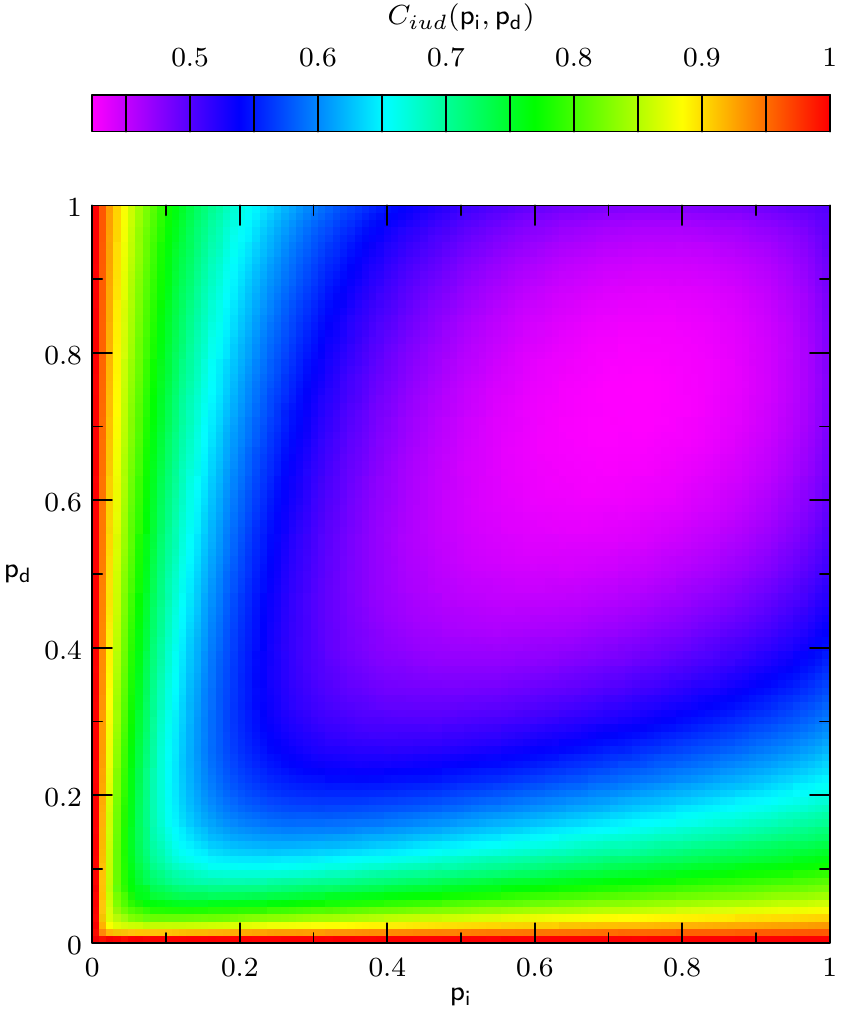}
\caption{The SIR $C_{iud}(\p_{\msi}, \p_{\msd})$  for the $W(\p_\msi, \p_\msd)$ channel.}
\label{fig_siraid}
\end{figure}
As a result of the memory in the $\mcl{Z}$ process, it is considerably harder than it was for the $W(\p)$ channel to obtain closed-form expressions for lower bounds on the capacity of $W(\p_\msi, \p_\msd)$ by computing information rates for known input distributions. However, the $W(\p_\msi, \p_\msd)$ channel is still an ergodic channel, and the SIR $C_{iud}(\p_{\msf{id}}, \p_{\msf{id}})$ can be obtained numerically. Fig.~\ref{fig_siraid} shows the contours of the SIR  for the $W(\p_\msi, \p_\msd)$ channel. Note that the symmetry proved in Proposition \ref{prop_1idsym} is evident.

\begin{figure}[!ht]
\centering
\includegraphics[width=.95\linewidth]{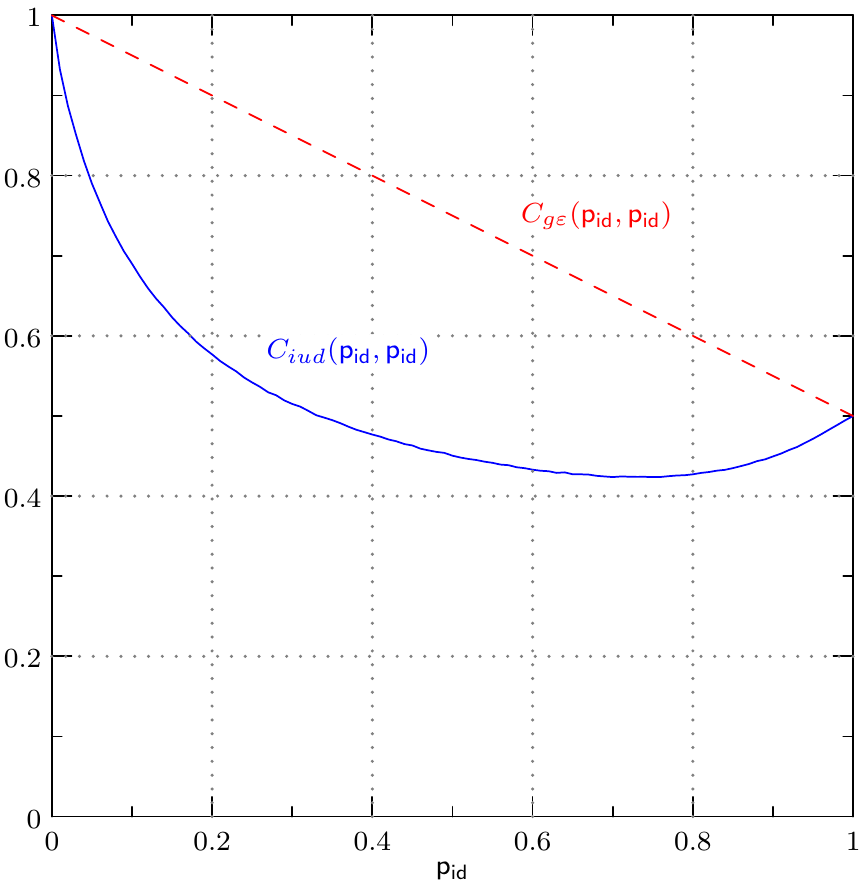}
\caption{The SIR $C_{iud}(\p_{\msf{id}}, \p_{\msf{id}})$ and the genie-erasure capacity $C_{g\e}(\p_{\msf{id}}, \p_{\msf{id}})$ for the symmetric binary Markov state channel $W(\p_{\msf{id}}, \p_{\msf{id}})$.}
\label{fig_insdelsir}
\end{figure}
In Fig.~\ref{fig_insdelsir}, we show the SIR for the case $\p_\msi = \p_\msd = \p_{\msf{id}}$, which we call the \emp{symmetric} binary Markov state channel $W(\p_{\msf{id}}, \p_{\msf{id}})$. The values of the SIR when $\p_{\msf{id}} = 0$ and $\p_{\msf{id}} = 1$ are easily explained. When $\p_{\msf{id}} = 0$, the channel is noiseless. When $\p_\msf{id} = 1$, the channel deterministically flips between the identity and the delayed channel so that every odd bit is repeated twice, and every even bit is lost, and the maximum achievable information transfer rate is $\frac{1}{2}$ bit per channel use. Also shown is the genie-erasure capacity in \eqref{eq_1idcapub}, which in this case becomes,
\[
C(\p_\msf{id}, \p_\msf{id}) \le C_{g\e}(\p_\msf{id}, \p_\msf{id}) = 1 - \frac{\p_\msf{id}}{2}.
\]
Interestingly enough, when $\p_\msf{id} = 1$, the SIR satisfies $C_{iud}(1, 1) = C_{g\e}(1, 1)$, so that $C(1, 1) = \frac{1}{2}$. We also include in Fig.~\ref{fig_oneinfsir} the SIR obtained for the channel $W(\p_\msi, 1)$, as well as the corresponding genie-erasure capacity, $C_{g\e}(\p_\msi, 1) = \frac{1}{1 + \p_\msi}$.
\begin{figure}[!ht]
\centering
\includegraphics[width=.95\linewidth]{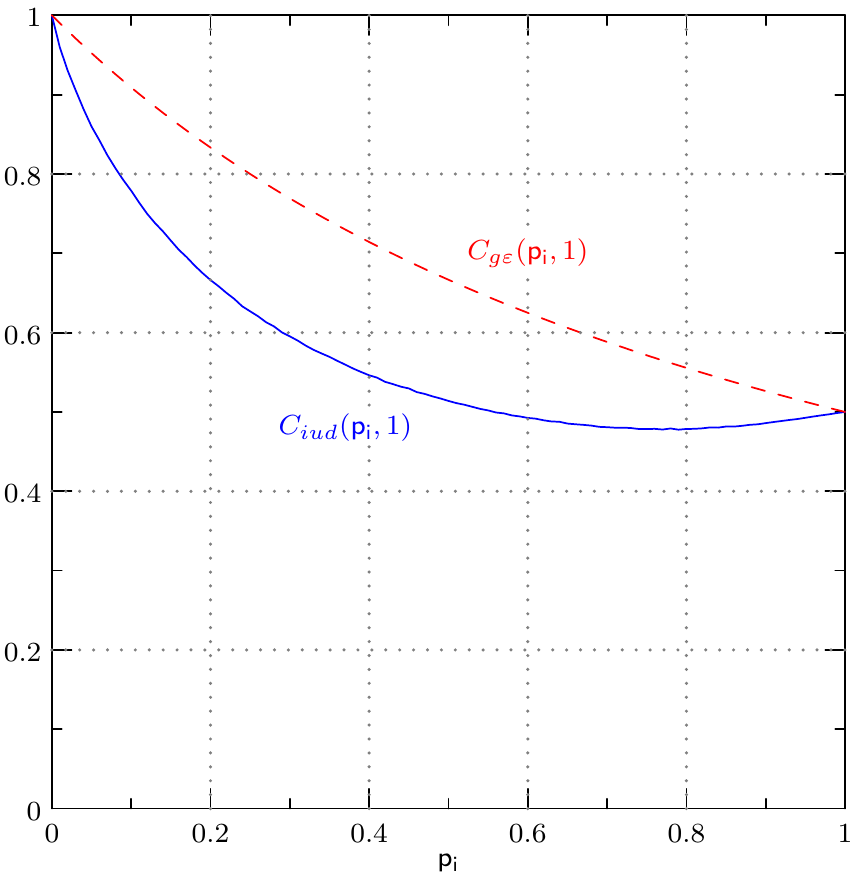}
\caption{The SIR $C_{iud}(\p_\msi, 1)$ and the genie-erasure capacity $C_{g\e}(\p_\msi, 1)$ for the channel $W(\p_\msi, 1)$ with the $\mcl{Z}$ process satisfying the $(1, \infty)$ constraint.}
\label{fig_oneinfsir}
\end{figure}

\section{$K$-ary Markov State Channel} \label{sec_kid}
We now consider a generalization of the binary Markov state channel as described in \eqref{eq_altchnmod}. Here, we allow $Z_i \in \mds{Z} = \{0, 1, \cdots, K - 1\}$ and let $\mcl{Z}$ be a first-order Markov process whose transition probabilities satisfy 
\begin{align}
\mbb{P}\{Z_i = z | Z_{i - 1} = z - 1\} &= \p_\msi, \notag \\
\mbb{P}\{Z_i = z | Z_{i - 1} = z + 1\} &= \p_\msd, \notag \\
\mbb{P}\{Z_i = z | Z_{i - 1} = z\} &= 1 - \p_\msi - \p_\msd \notag
\end{align}
for $z \in \{1, 2, \cdots, K - 2\}$. Further, we have
\[
\mbb{P}\{Z_i = 1 | Z_{i - 1} = 0\} = \p_\msi = 1 - \mbb{P}\{Z_i = 0 | Z_{i - 1} = 0\}
\]
and
\begin{align}
\mbb{P}\{Z_i = K - 2 &| Z_{i - 1} = K - 1\} = \p_\msd \notag \\
&= 1 - \mbb{P}\{Z_i = K - 1 | Z_{i - 1} = K - 1\}. \notag
\end{align}
We denote this by $\mcl{Z} \sim \mcl{M}^{(K)}_1(\p_\msi, \p_\msd)$. Note that when $K = 2$, this model gives the binary Markov state channel considered earlier. We will hence be interested in the $K$-ary Markov state channel where $K > 2$, which we denote by $W(\p_\msi, \p_\msd, K)$. Note that this channel now generalizes the binary Markov state channel in the sense that it allows up to $(K - 1)$ consecutive insertions or deletions.

We further assume that the parameters $\p_\msi, \p_\msd$ are chosen such that the process $\mcl{Z}$ is aperiodic and irreducible so that $\mcl{Y}$ is ergodic. The channel space is given by  $\p_\msi \in [0, 1], \p_\msd \in [0, 1 - \p_\msi]$. The channel symmetry argument of Proposition \ref{prop_1idsym} holds in this case also, so that
\[
C(\p_\msi, \p_\msd, K) = C(\p_\msd, \p_\msi, K)
\]
where $C(\p_\msi, \p_\msd, K)$ is the capacity of $K$-ary Markov state channel $W(\p_\msi, \p_\msd, K)$. Hence, the channel space can be further reduced to $\p_\msi \in [0, 1], \p_\msd \in [0, \min\{\p_\msi, 1 - \p_\msi\}]$. As was the case for the binary Markov state channel, we have $C(0, \p_\msd, K) = C(\p_\msi, 0, K) = 1{\ }\forall{\ }\p_\msi, \p_\msd \in [0, 1]$.

The genie-erasure capacity of the $W(\p_\msi, \p_\msd, K)$ channel is given by the capacity of a correlated erasure channel with an average erasure rate
\begin{align}
\mbb{P}&\{Z_i = Z_{i - 1} - 1\} = \sum_{z = 1}^{K - 1} \mbb{P}\{Z_{i - 1} = z, Z_i = z - 1\} \notag \\
&= \sum_{z = 1}^{K - 1} \mbb{P}\{Z_{i - 1} = z\}\mbb{P}\{Z_i = z - 1 | Z_{i - 1} = z\} = \sum_{z = 1}^{K - 1} \pi_z\p_\msd \notag \\
&= \p_\msd\left(1 - \left(\frac{\p_\msd}{\p_\msi}\right)^{K - 1}\left(\frac{1 - \left(\frac{\p_\msd}{\p_\msi}\right)}{1 - \left(\frac{\p_\msd}{\p_\msi}\right)^K}\right)\right) = \p_{\e} \notag
\end{align}
where $\pi_z = \left(\frac{\p_\msd}{\p_\msi}\right)^{K - 1 - z}\left(\frac{1 - \left(\frac{\p_\msd}{\p_\msi}\right)}{1 - \left(\frac{\p_\msd}{\p_\msi}\right)^K}\right)$ is the steady state probability of $Z_i = z$, so that
\[
C(\p_\msi, \p_\msd, K) \le 1 - \p_{\e} \stackrel{\Delta}{=} C_{g\e}(\p_\msi, \p_\msd, K).
\]
For the \emp{symmetric} $K$-ary Markov state channel, we have
\[
C(\p_\msf{id}, \p_\msf{id}, K) \le 1 - \p_\msf{id}\frac{K - 1}{K} = C_{g\e}(\p_\msf{id}, \p_\msf{id}, K)
\]
for $\p_\msf{id} \in [0, \frac{1}{2}]$ because in this case, $\pi_z = \frac{1}{K}{\ }\forall{\ }z\in\{0, 1, \cdots, K - 1\}$. Note that $C_{g\e}(\p_\msf{id}, \p_\msf{id}, K)$ reduces to the capacity of a BEC with erasure probability $\p_\msf{id}$, $C_{BEC}(\p_\msf{id}) = 1 - \p_\msf{id}$, when $K \tends \infty$.

The SIR of the $W(\p_\msi, \p_\msd, K)$ channel can be obtained numerically as in the case of the $W(\p_\msi, \p_\msd)$ channel. However, the computational complexity of the BCJR algorithm \cite{sor_07_tit_sir} used to estimate the SIR increases roughly exponentially in the size $K$ of the alphabet of the process $\mcl{Z}$, and hence the evaluation of the SIR for  $K>2$ may require considerable computing resources. We carried out the calculation for the case when $K = 3$, and the estimated SIR is shown in Fig.~\ref{fig_2insdelsir}.
\begin{figure}[!ht]
\centering
\includegraphics[width=.95\linewidth]{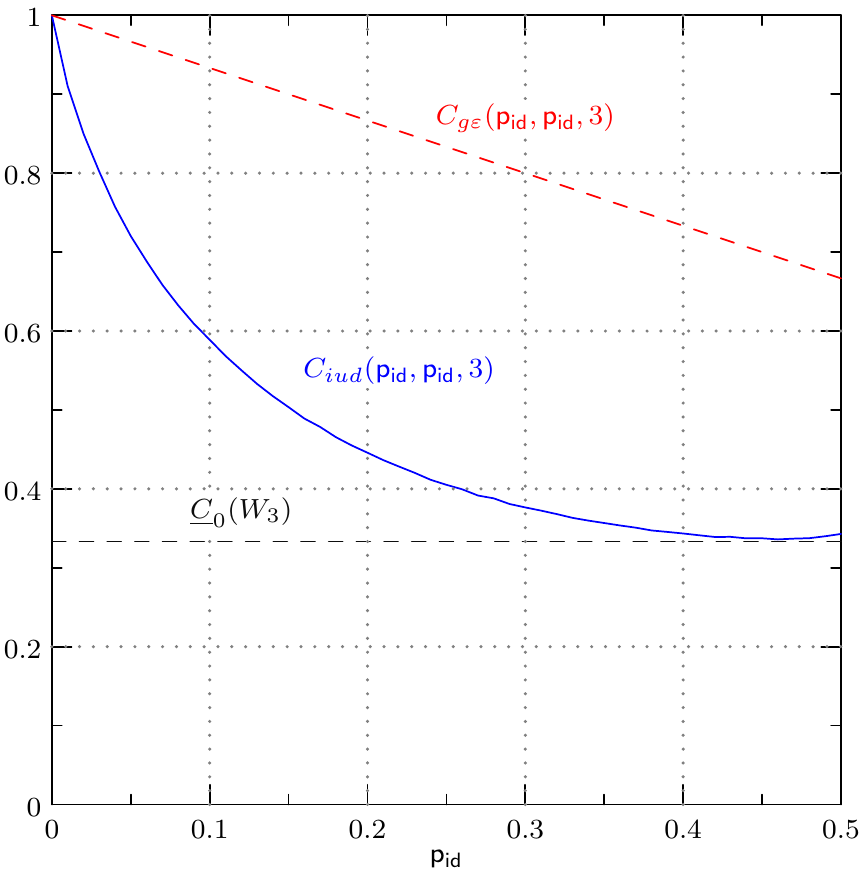}
\caption{The SIR $C_{iud}(\p_\msf{id}, \p_\msf{id}, 3)$ for the symmetric $3$-ary Markov state channel, along with the genie-erasure capacity $C_{g\e}(\p_\msf{id}, \p_\msf{id}, 3)$. Also shown is the lower bound on the zero-error capacity $\underline{C}_0(W_3)$ from Corollary \ref{cor_zeck}.}
\label{fig_2insdelsir}
\end{figure}

\section{Zero-error capacity} \label{sec_zec}
The \emp{zero-error capacity} $C_0$ of a discrete-memoryless channel was introduced by Shannon \cite{sha_56_tit_zec}, from which it is readily seen that a noisy discrete-output binary-input memoryless channel has a zero-error capacity of $0$, i.e., no information can be transmitted over such a channel with zero error. However, we will now show that the zero-error capacity of the noisy discrete-output binary-input channel in \eqref{eq_chnmod} is strictly positive. We will denote the generic channel in  \eqref{eq_chnmod} for all binary $\mcl{Z}$ processes by $W$.

\begin{prop} \label{prop_zec}
$C_0(W) = \frac{1}{2}$.
\end{prop}
\begin{proof}
Consider the $2$-repetition code. The channel input process $\mcl{X}$ and the message process $\msc{M}$ satisfy the relationship $X_{2i-1} = M_i, X_{2i} = M_i$ for $i = 1, 2, \cdots$. Since $X_{2i-1} = X_{2i}$, $Y_{2i} = X_{2i} = M_i$ so that discarding the $Y_{2i-1}$s gives us $\msc{M}$ exactly, thereby achieving zero error at a rate $\frac{1}{2}$. Thus, $C_0(W) \ge \frac{1}{2}$.

It is obvious that the capacity of a channel is an upper bound on the zero-error capacity. From \eqref{eq_1idcapub}, we can see that $C(1, 1) \le \frac{1}{2}$, and the coding scheme above achieves a rate $0.5$ for any realization of $\mcl{Z}$. Thus, $C(1, 1) = \frac{1}{2}$. Therefore, as long as the channel admits the infinitely long alternating sequence $\mcl{Z} = 010101\cdots$, we have $C_0(W) = \frac{1}{2}$.

Note that both the Bernoulli state channel $W(\p)$ and the binary Markov state channel $W(\p_\msi, \p_\msd)$ can generate the infinite sequence $\mcl{Z} = 010101\cdots$, albeit with vanishing probability. It is also worth noting here that this infinite sequence satisfies the $(1, \infty)$ constraint (See Section \ref{ssec_tdmr}).
\end{proof}
Thus, unlike binary-input discrete memoryless channels, the binary-input discrete channel in \eqref{eq_chnmod} allows a non-zero information rate even under the severe requirement of zero-error. From this result and Fig.~\ref{fig_iidbsscap}, it is clear that even under the milder condition of asymptotically vanishing error probability, random linear coset coding achieves a lower rate than the $2$-repetition code, which guarantees zero errors, for the Bernoulli state channel over a range of $\p$ values. Referring to Fig.~\ref{fig_siraid}, the same can be said for the binary Markov state channels. However, by using a first-order Markov input (cf. Fig. \ref{fig_iidbsscap}), higher rates than that of the zero-error achieving scheme are achievable for all Bernoulli state channels, although only with asymptotically vanishing error probability.

The zero-error capacity of the $W(\p_\msi, \p_\msd, K)$ channel, denoted by $C_0(W_K)$, satisfies the following bounds.
\begin{cor} \label{cor_zeck}
$\frac{1}{K} \le C_0(W_K) \le \frac{K + 1}{2K}$.
\end{cor}
\begin{proof}
Repeating every bit $K$ times achieves zero error, since every $K^{\text{th}}$ bit is always correct for any realization of the $\mcl{Z}$ process. Hence, $C_0(W_K) \ge \frac{1}{K} = \underline{C}_0(W_K)$.

The smallest upper bound for the capacity is also an upper bound for the zero-error capacity. Thus,
\begin{align}
C_0(W_K) &\le \min_{\p_\msi, \p_\msd}C(\p_\msi, \p_\msd, K) \notag \\
&\le \min_{\p_\msi, \p_\msd}C_{g\e}(\p_\msi, \p_\msd, K) \notag \\
&= C_{g\e}(\frac{1}{2}, \frac{1}{2}, K) = \frac{K + 1}{2K}. \notag
\end{align}
Observe that when $K = 2$, the minimum genie-erasure capacity occurs at $\p_\msi = \p_\msd = 1$ and at this point the upper bound $C_{g\e}(1, 1, 2)$ is same as the lower bound $\underline{C}_0(W_2)$ which was used in the proof of Proposition \ref{prop_zec}.
\end{proof}
Note that the bounds proposed above are loose asymptotically, i.e., $\frac{1}{K} \tends 0$ and $\frac{K + 1}{2K} \tends \frac{1}{2}$ as $K \tends \infty$.

\section{Conclusions} \label{sec_conc}
We proposed a new write channel model for bit-patterned media recording that reflects the data dependence of write synchronization errors. This model generates both substitution-like errors and insertion-deletion errors whose statistics are determined by an underlying channel state process. For Bernoulli and Markov channel state models, we studied information-theoretic limits imposed by the channel, computing bounds and numerical estimates for the maximum achievable information rate under different assumptions on the channel input statistics. We then generalized the Markov channel state model to allow a channel state space of size $K>2$ and computed the SIR numerically for the case $K=3$.
Finally, we showed that the rate-$\frac{1}{2}$ 2-repetition code achieves the zero-error capacity over the new write channel when the channel state space is binary. Bounds on the zero-error capacity of the general $K$-ary Markov channel state model were also presented. In future work, we plan to combine the new write channel model in cascade with the partial-response readback channel model and investigate the achievable rates, as was done for the i.i.d. insertion-deletion channel in \cite{hu_10_tcom_insdel}. Furthermore, the channel model considered here can be readily represented using a factor graph and hence the construction of sparse graph-based codes with iterative decoding techniques seems possible.

\twobibs{
% \bibliographystyle{IEEEtran}
% \bibliography{/media/Acads/work/Bib/mybib}

\begin{thebibliography}{10}
\providecommand{\url}[1]{#1}
\csname url@samestyle\endcsname
\providecommand{\newblock}{\relax}
\providecommand{\bibinfo}[2]{#2}
\providecommand{\BIBentrySTDinterwordspacing}{\spaceskip=0pt\relax}
\providecommand{\BIBentryALTinterwordstretchfactor}{4}
\providecommand{\BIBentryALTinterwordspacing}{\spaceskip=\fontdimen2\font plus
\BIBentryALTinterwordstretchfactor\fontdimen3\font minus
  \fontdimen4\font\relax}
\providecommand{\BIBforeignlanguage}[2]{{%
\expandafter\ifx\csname l@#1\endcsname\relax
\typeout{** WARNING: IEEEtran.bst: No hyphenation pattern has been}%
\typeout{** loaded for the language `#1'. Using the pattern for}%
\typeout{** the default language instead.}%
\else
\language=\csname l@#1\endcsname
\fi
#2}}
\providecommand{\BIBdecl}{\relax}
\BIBdecl

\bibitem{iye_10_isit_wcm}
A.~R. Iyengar, P.~H. Siegel, and J.~K. Wolf, ``Data-dependent write channel
  model for magnetic recording,'' in \emph{Proc. IEEE Int. Symp. Inf. Theory},
  Austin, TX, USA, Jun. 13-18, 2010, pp. 958--962.

\bibitem{kob_71_tcomtech_codmagrec}
H.~Kobayashi, ``A survey of coding schemes for transmission or recording of
  digital data,'' \emph{IEEE Trans. on Comm. Techn.}, vol.~19, no.~6, pp.
  1087--1100, Dec. 1971.

\bibitem{kab_75_tcom_prs}
P.~Kabal and S.~Pasupathy, ``Partial-response signaling,'' \emph{IEEE Trans.
  Commun.}, vol.~23, no.~9, pp. 921--934, Sep. 1975.

\bibitem{hu_07_tmag_bpmwc}
J.~Hu, T.~Duman, E.~Kurtas, and M.~Erden, ``Bit-patterned media with written-in
  errors: modeling, detection, and theoretical limits,'' \emph{IEEE Trans.
  Magn.}, vol.~43, no.~8, pp. 3517--3524, Aug. 2007.

\bibitem{iye_09_all_bpm}
A.~R. Iyengar, P.~H. Siegel, and J.~K. Wolf, ``{LDPC} codes for the cascaded
  {BSC-BAWGN} channel,'' in \emph{Proc. 47th Annual Allerton Conf. on
  Communication, Control and Computing}, Sep. 30 - Oct. 2, 2009, pp. 620--627.

\bibitem{whi_97_tmag_bpm}
R.~White, R.~Newt, and R.~Pease, ``Patterned media: a viable route to 50
  {G}bit/in$^2$ and up for magnetic recording?'' \emph{IEEE Trans. Magn.},
  vol.~33, no.~1, pp. 990--995, Jan 1997.

\bibitem{liv_09_tmag_wwz}
B.~Livshitz, A.~Inomata, H.~Bertram, and V.~Lomakin, ``Semi-analytical approach
  for analysis of {BER} in conventional and staggered bit patterned media,''
  \emph{IEEE Trans. Magn.}, vol.~45, no.~10, pp. 3519--3522, Oct. 2009.

\bibitem{mit_08_tit_sticky}
M.~Mitzenmacher, ``Capacity bounds for sticky channels,'' \emph{IEEE Trans.
  Inf. Theory}, vol.~54, no.~1, pp. 72--77, Jan. 2008.

\bibitem{dig_01_all_delchn}
S.~Diggavi and M.~Grossglauser, ``On transmission over deletion channels,'' in
  \emph{Proc. 39th Annual Allerton Conf. on Communication, Control and
  Computing}, Oct. 3-5, 2001, pp. 573--582.

\bibitem{dig_06_tit_finbufchn}
------, ``On information transmission over a finite buffer channel,''
  \emph{IEEE Trans. Inf. Theory}, vol.~52, no.~3, pp. 1226--1237, Mar. 2006.

\bibitem{dri_06_tit_lbcapdc}
E.~Drinea and M.~Mitzenmacher, ``On lower bounds for the capacity of deletion
  channels,'' \emph{IEEE Trans. Inf. Theory}, vol.~52, no.~10, pp. 4648--4657,
  Oct. 2006.

\bibitem{mit_06_tit_lbcapdc}
M.~Mitzenmacher and E.~Drinea, ``A simple lower bound for the capacity of the
  deletion channel,'' \emph{IEEE Trans. Inf. Theory}, vol.~52, no.~10, pp.
  4657--4660, Oct. 2006.

\bibitem{dig_07_isit_capubdc}
S.~Diggavi, M.~Mitzenmacher, and H.~D. Pfister, ``Capacity upper bounds for the
  deletion channel,'' in \emph{Proc. IEEE Int. Symp. Inf. Theory}, Nice,
  France, Jun. 24-29, 2007, pp. 1716--1720.

\bibitem{fer_10_tit_bdccap}
D.~Fertonani and T.~Duman, ``Novel bounds on the capacity of the binary
  deletion channel,'' \emph{IEEE Trans. Inf. Theory}, vol.~56, no.~6, pp. 2753
  --2765, Jun. 2010.

\bibitem{kir_10_tit_ddccap}
A.~Kirsch and E.~Drinea, ``Directly lower bounding the information capacity for
  channels with i.i.d. deletions and duplications,'' \emph{IEEE Trans. Inf.
  Theory}, vol.~56, no.~1, pp. 86 --102, Jan. 2010.

\bibitem{hu_10_tcom_insdel}
J.~Hu, T.~Duman, M.~Erden, and A.~Kavcic, ``Achievable information rates for
  channels with insertions, deletions, and intersymbol interference with i.i.d.
  inputs,'' \emph{IEEE Trans. Commun.}, vol.~58, no.~4, pp. 1102 --1111, Apr.
  2010.

\bibitem{woo_09_tmag_hdmr}
R.~Wood, M.~Williams, A.~Kavcic, and J.~Miles, ``The feasibility of magnetic
  recording at 10 {T}erabits per square inch on conventional media,''
  \emph{IEEE Trans. Magn.}, vol.~45, no.~2, pp. 917--923, Feb. 2009.

\bibitem{imm_99_bok_codsto}
K.~A.~S. Immink, \emph{Codes for {M}ass {D}ata {S}torage {S}ystems}.\hskip 1em
  plus 0.5em minus 0.4em\relax Shannon Foundation Publishers, The Netherlands,
  1999.

\bibitem{maz_10_isit_hdmr}
A.~Mazumdar, A.~Barg, and N.~Kashyap, ``Coding for high-density magnetic
  recording,'' in \emph{Proc. IEEE Int. Symp. Inf. Theory}, Austin, TX, USA,
  Jun. 13-18, 2010, pp. 978--982.

\bibitem{cov_06_bok_infoth}
T.~M. Cover and J.~A. Thomas, \emph{Elements of Information Theory}.\hskip 1em
  plus 0.5em minus 0.4em\relax 2nd ed. New York: John Wiley and Sons, 2006.

\bibitem{gal_68_bok_infoth}
R.~G. Gallager, \emph{Information Theory and Reliable Communication}.\hskip 1em
  plus 0.5em minus 0.4em\relax New York: John Wiley and Sons, 1968.

\bibitem{ric_08_bok_mct}
T.~Richardson and R.~Urbanke, \emph{Modern Coding Theory}.\hskip 1em plus 0.5em
  minus 0.4em\relax Cambridge University Press, 2008.

\bibitem{bah_74_tit_bcjr}
L.~Bahl, J.~Cocke, F.~Jelinek, and J.~Raviv, ``Optimal decoding of linear codes
  for minimizing symbol error rate (corresp.),'' \emph{IEEE Trans. Inf.
  Theory}, vol.~20, no.~2, pp. 284--287, Mar. 1974.

\bibitem{kav_01_glb_marcap}
A.~Kavcic, ``On the capacity of {M}arkov sources over noisy channels,'' in
  \emph{Proc. IEEE Globecom}, vol.~5, San Antonio, TX, USA, Nov. 25-29, 2001,
  pp. 2997--3001.

\bibitem{von_08_tit_gbaa}
P.~Vontobel, A.~Kavcic, D.~Arnold, and H.-A. Loeliger, ``A generalization of
  the {B}lahut-{A}rimoto algorithm to finite-state channels,'' \emph{IEEE
  Trans. Inf. Theory}, vol.~54, no.~5, pp. 1887 --1918, May 2008.

\bibitem{pap_91_bok_prob}
A.~Papoulis, \emph{Probability, Random Variables and Stochastic
  Processes}.\hskip 1em plus 0.5em minus 0.4em\relax New York: McGraw-Hill
  Inc., 1991.

\bibitem{sor_07_tit_sir}
J.~B. Soriaga, H.~D. Pfister, and P.~H. Siegel, ``Determining and approaching
  achievable rates of binary intersymbol interference channels using multistage
  decoding,'' \emph{IEEE Trans. Inf. Theory}, vol.~53, no.~4, pp. 1416--1429,
  Apr. 2007.

\bibitem{sha_56_tit_zec}
C.~Shannon, ``The zero error capacity of a noisy channel,'' \emph{IRE Trans.
  Inf. Theory}, vol.~2, no.~3, pp. 8--19, Sep. 1956.

\end{thebibliography}
}
{

}

\end{document}